\documentclass[12pt,a4paper]{article}
\pdfoutput=1

\usepackage[utf8]{inputenc}
\usepackage[T1]{fontenc}
\usepackage{amsmath,amssymb,amsthm}
\usepackage{mathtools}
\usepackage{graphicx}
\usepackage[margin=1in]{geometry}
\usepackage{natbib}
\usepackage{booktabs}
\usepackage{caption}
\usepackage{subcaption}
\usepackage{hyperref}
\usepackage{cleveref}
\usepackage{float}
\usepackage{setspace}
\usepackage{xcolor}
\usepackage{enumitem}

\newtheorem{theorem}{Theorem}
\newtheorem{proposition}[theorem]{Proposition}
\newtheorem{corollary}[theorem]{Corollary}

\theoremstyle{definition}
\newtheorem{remark}{Remark}
\newtheorem{definition}{Definition}
\newtheorem{example}{Example}

\newcommand{\E}{\mathbb{E}}
\newcommand{\Q}{\mathbb{Q}}
\newcommand{\R}{\mathbb{R}}

\newcommand{\diff}{\mathrm{d}}
\newcommand{\ee}{\mathrm{e}}
\newcommand{\eTAO}{e_{\mathrm{TAO}}}

\hypersetup{
    colorlinks=true,
    linkcolor=blue!70!black,
    citecolor=blue!70!black,
    urlcolor=blue!70!black
}

\title{\textbf{Option Pricing on Automated Market Maker Tokens}}

\author{Philip Z.\ Maymin\thanks{Fairfield University Dolan School of Business. Email: \texttt{pmaymin@fairfield.edu}.}}

\date{February 2026}

\begin{document}

\maketitle

\begin{abstract}
We derive the stochastic price process for tokens whose sole price discovery mechanism is a constant-product automated market maker (AMM). When the net flow into the pool follows a diffusion, the token price follows a constant elasticity of variance (CEV) process, nesting Black--Scholes as the limiting case of infinite liquidity. We obtain closed-form European option prices and introduce liquidity-adjusted Greeks. The CEV structure generates a leverage effect---volatility rises as price falls---whose normalized implied volatility skew depends only on the pool's weighting parameter, not on pool depth: Black--Scholes underprices 20\%-out-of-the-money puts by roughly 6\% in implied volatility terms at every pool depth, while the absolute pricing discrepancy vanishes as pools deepen. Empirically, after controlling for pool depth and flow volatility, realized return variance across 90 Bittensor subnets exhibits a strongly negative price elasticity, decisively rejecting geometric Brownian motion and consistent with the CEV prediction. A complementary delta-hedged backtest across 82 subnets confirms near-identical hedging errors at the money, consistent with the prediction that pricing differences are concentrated in the wings.
\end{abstract}



\clearpage

\section{Introduction}\label{sec:intro}

Automated market makers (AMMs) have become the dominant mechanism for token exchange in decentralized finance (DeFi), processing hundreds of billions of dollars in cumulative trading volume since the launch of Uniswap in 2018 \citep{adams2021uniswap}. Unlike traditional order-book exchanges where prices emerge from the interaction of discrete buy and sell orders, AMMs determine prices algorithmically from the ratio of token reserves held in liquidity pools. The most widely adopted design, the \emph{constant-product} AMM, maintains the invariant $x \cdot y = k$ across two token reserves $x$ and $y$, with the marginal exchange rate given by the reserve ratio \citep{angeris2020improved}.

A new class of AMM-native tokens has emerged in which the AMM is not merely a trading venue but the \emph{sole price discovery mechanism}. There is no order book, no off-chain market, and no external price oracle; the bonding curve \emph{is} the market. The most prominent example is the Bittensor network's Dynamic TAO (dTAO) system, launched in February 2025, which assigns each of its subnets an ``alpha'' token traded exclusively through a dedicated constant-product AMM against the network's native TAO token \citep{bittensor2025dtao}. With over 60 active subnets, a total staked value exceeding \$3 billion, and a growing ecosystem of subnet operators seeking to hedge treasury exposure, pricing derivatives on these tokens is an increasingly practical concern. Bittensor's setting also provides the cleanest possible test environment for AMM price dynamics: because there is no competing venue, the observed price process is entirely determined by the AMM mechanics, making it possible to test the theory without confounding effects from external price discovery.

The Black--Scholes model \citep{black1973pricing} assumes the underlying follows geometric Brownian motion (GBM) with constant volatility, an assumption justified when price changes are driven by exogenous information arrival on a deep, frictionless market. For AMM-native tokens, the price is \emph{endogenously determined} by the bonding curve: every trade mechanically shifts the reserve ratio, and the resulting price dynamics inherit the nonlinear structure of the AMM itself.

The central contribution of this paper is to derive the stochastic process governing AMM token prices from first principles and to develop the resulting option pricing framework. Our main result (\Cref{thm:main}) establishes a simple identity: when the net staking flow into an AMM pool follows a Brownian diffusion, the token price follows a \emph{constant elasticity of variance} (CEV) model \citep{cox1975notes,cox1996constant} with exponent $\beta = w$, the weight of the numeraire token in the pool. For the standard constant-product AMM, $\beta = 1/2$. This is not an empirical estimate; it is a mathematical consequence of the bonding curve.

This result has several important implications:
\begin{enumerate}[leftmargin=*]
    \item \textbf{Black--Scholes as a limiting case.} As pool liquidity $k \to \infty$, the CEV volatility parameter vanishes and the price becomes deterministic. For large but finite $k$, the process approximates GBM, and Black--Scholes applies with an effective volatility that is inversely proportional to the square root of pool depth. The result gives a precise characterization of when standard pricing models are adequate.

    \item \textbf{Liquidity-dependent volatility and the leverage effect.} The instantaneous volatility of the AMM token price is $\sigma(P) = \delta P^{\beta - 1}$, where $\delta$ is proportional to the flow volatility and inversely proportional to pool depth. For $\beta = 1/2$, volatility decreases with the square root of price. This produces a structural leverage effect (negative correlation between price and volatility) that is a consequence of the bonding curve mechanics, not capital structure.

    \item \textbf{AMM-specific Greeks.} Beyond the standard sensitivities, we derive a ``liquidity Greek'' $\Lambda = \partial V / \partial k$ measuring option value sensitivity to pool depth, and an ``emission Greek'' $\mathcal{E} = \partial V / \partial e$ capturing sensitivity to the token emission rate.

    \item \textbf{Quantifiable pricing discrepancy.} We provide closed-form expressions for the pricing discrepancy relative to Black--Scholes as a function of pool depth and moneyness (\Cref{fig:price_vs_k,fig:vol_smile}), enabling practitioners to assess when the standard model is inadequate and by how much.

    \item \textbf{Universal implied volatility skew.} We show that the normalized implied volatility skew depends only on $\beta$, not on the volatility parameter $\delta$ or pool depth $k$ (\Cref{prop:skew}). For $\beta = 1/2$, Black--Scholes underprices 20\%-out-of-the-money puts by roughly 6\% in implied volatility terms at every pool depth. This is a structural, falsifiable prediction.

    \item \textbf{Empirical validation.} The CEV model predicts that, after controlling for pool depth and flow volatility, realized return variance should scale as $P^{2(\beta-1)} = P^{-1}$ for constant-product AMMs. A cross-sectional test across 90 Bittensor subnets yields a median variance elasticity of $-0.86$ (interquartile range $[-0.98, -0.71]$), strongly rejecting the GBM null of zero ($p < 0.0001$) and consistent with $\beta$ near $1/2$.
\end{enumerate}

Our work connects three strands of literature. The first is the financial theory of AMMs, including analyses of impermanent loss \citep{loesch2021impermanent}, AMM design \citep{angeris2020improved,angeris2022optimal}, and the relationship between AMM liquidity provision and option payoffs \citep{clark2020replicating,guillaume2024unified}. The second is the CEV option pricing literature initiated by \citet{cox1975notes}, with closed-form solutions developed by \citet{schroder1989computing} and subsequent extensions \citep{davydov2003pricing,larguinho2013computation}. The third is the nascent literature on Bittensor's tokenomics and decentralized AI markets \citep{bittensor2025dtao}.

The remainder of the paper is organized as follows. \Cref{sec:related} surveys the related literature. \Cref{sec:background} describes the institutional setting, with particular attention to Bittensor's dTAO mechanism. \Cref{sec:model} derives the CEV price process from AMM fundamentals. \Cref{sec:pricing} develops the option pricing framework, including closed-form solutions, Greeks, and emission extensions. \Cref{sec:numerical} provides numerical analysis, Monte Carlo validation, calibration to Bittensor data, and empirical tests of the CEV variance elasticity. \Cref{sec:discussion} discusses limitations and extensions. \Cref{sec:conclusion} concludes.

\section{Related Literature}\label{sec:related}

\paragraph{AMM theory.} The mathematical foundations of constant-function market makers were established by \citet{angeris2020improved}, who characterized the set of feasible trades and the connection between AMM prices and external market prices. \citet{angeris2022optimal} extended this to optimal routing across multiple pools. \citet{milionis2022automated} introduced the loss-versus-rebalancing (LVR) framework, decomposing LP returns into market risk and a predictable adverse selection cost, providing what they term a ``Black--Scholes formula for AMMs.'' \citet{park2023conceptual} identifies conceptual flaws in constant-product pricing, including persistent arbitrage and front-running vulnerabilities. The welfare properties and fee structures of AMMs are analyzed by \citet{roughgarden2024transaction}, while \citet{cartea2023decentralised} develop a continuous-time model of AMM dynamics incorporating arbitrageurs and informed traders.

\paragraph{AMMs and options.} \citet{clark2020replicating} showed that a constant-product AMM liquidity provider is effectively writing a perpetual straddle, and \citet{loesch2021impermanent} formalized impermanent loss as a variance-dependent cost. \citet{hasbrouck2024economic} formalize concentrated liquidity positions as covered calls, showing that LPs forgo option time value in exchange for fees. \citet{fukasawa2023weighted} prove that impermanent loss can be hedged with weighted variance swaps, connecting AMM positions to gamma swaps. \citet{guillaume2024unified} developed static hedging strategies using vanilla options, and \citet{theammbook2022} used the Black--Scholes formula to estimate divergence loss magnitudes. \citet{bichuch2024derivative} derive risk-neutral prices and Greeks for LP tokens themselves, treating the LP position as a derivative of the underlying asset prices under GBM. On-chain options protocols have motivated further pricing theory: \citet{dave2023perpetual} and \citet{blockscholes2025} analyze perpetual options in DeFi, while \citet{singh2025options} survey the broader DeFi options ecosystem.

\paragraph{CEV models.} The constant elasticity of variance model was introduced by \citet{cox1975notes} and extended by \citet{cox1996constant}. Empirical estimation and testing of the CEV model against equity option data were provided by \citet{beckers1980cev} and \citet{emanuel1982further}. Closed-form European option prices were derived by \citet{schroder1989computing} using the non-central chi-squared distribution. \citet{davydov2003pricing} provided eigenfunction-based pricing, and \citet{larguinho2013computation} improved the numerical stability of the CEV formula. The CEV model generates implied volatility skew controlled by the elasticity parameter, making it a parsimonious alternative to stochastic volatility models \citep{heston1993closed}. In the CEV literature, the elasticity parameter $\beta$ is typically an empirical quantity estimated from data; our contribution is to show that for AMM tokens, $\beta$ is pinned by the pool design ($\beta = w$) rather than estimated.

\paragraph{Concurrent work.} In independent and concurrent work, \citet{hitier2025dynamics} models LP portfolio value in constant-product AMMs under GBM. That paper assumes the external asset price follows GBM and derives LP value dynamics, whereas we derive the \emph{endogenous} price process of a token that trades \emph{only} through the AMM. The approaches are complementary: Hitier's applies when the AMM is one of many trading venues; ours applies when the AMM is the sole price discovery mechanism.

\paragraph{Our contribution.} The existing literature treats LP positions \emph{as} derivatives (of the underlying asset price) and prices them accordingly \citep{bichuch2024derivative,hasbrouck2024economic}. We work in the opposite direction: we derive the stochastic process governing the \emph{token price itself} from the bonding curve mechanics, and then price derivatives \emph{on} that token. The result, that AMM token prices follow a CEV process with exponent equal to the pool weight, provides this missing link and yields a complete option pricing framework.

\section{Institutional Background}\label{sec:background}

\subsection{Constant-Product Automated Market Makers}

A constant-product AMM maintains two token reserves $(x, y)$ subject to the invariant
\begin{equation}\label{eq:cpamm}
    x \cdot y = k,
\end{equation}
where $k > 0$ is constant during any individual swap. A trader who deposits $\Delta x$ units of token $X$ receives $\Delta y$ units of token $Y$, determined by
\begin{equation}\label{eq:swap}
    (x + \Delta x)(y - \Delta y) = k \quad \Longrightarrow \quad \Delta y = \frac{y \cdot \Delta x}{x + \Delta x}.
\end{equation}
The marginal price of token $Y$ in terms of token $X$ is
\begin{equation}\label{eq:marginal_price}
    P = \frac{x}{y},
\end{equation}
obtained by differentiating the invariant. Equation~\eqref{eq:swap} implies that the effective execution price for a trade of size $\Delta x$ exceeds the marginal price by a slippage term of order $\Delta x / x$, making large trades progressively more expensive. The design was introduced by \citet{buterin2017path} and formalized by \citet{angeris2020improved}.

\begin{example}[A simple constant-product AMM]\label{ex:simple_amm}
Suppose a pool holds $x = 1{,}000$ TAO and $y = 40{,}000$ ALPHA, so $k = x \cdot y = 4 \times 10^7$ and the marginal price is $P = 1{,}000/40{,}000 = 0.025$ TAO per ALPHA. A trader who stakes 100 TAO receives
\[
    \Delta y = \frac{40{,}000 \times 100}{1{,}000 + 100} = 3{,}636.4 \;\text{ALPHA}.
\]
After the trade the reserves are $x = 1{,}100$ and $y = 36{,}363.6$, the invariant is still $k = 4 \times 10^7$, and the new price is $P' = 1{,}100/36{,}363.6 \approx 0.0302$. A deposit equal to 10\% of the TAO reserve moved the price by 21\%. This amplification of flows into price changes is the mechanism that generates the CEV dynamics derived in \Cref{sec:model}. Note that a pool ten times deeper ($k = 4 \times 10^9$) would produce correspondingly smaller price impact from the same trade, foreshadowing the Black--Scholes limit of infinite pool depth.
\end{example}

\subsection{Generalized Constant-Function AMMs}

The constant-product design is a special case of the \emph{constant-weighted-product} family:
\begin{equation}\label{eq:cwp}
    x^w \cdot y^{1-w} = K,
\end{equation}
where $w \in (0,1)$ is the weight of token $X$ (the numeraire) and $K > 0$. For $w = 1/2$, this reduces to the constant-product AMM with $K = \sqrt{k}$. Platforms such as Balancer implement arbitrary weights, enabling asymmetric exposure \citep{martinelli2019balancer}. The marginal price under \eqref{eq:cwp} is
\begin{equation}\label{eq:general_price}
    P = \frac{1-w}{w} \cdot \frac{x}{y}.
\end{equation}

\subsection{Bittensor and Dynamic TAO}\label{sec:bittensor}

Bittensor is a decentralized network for AI services organized into \emph{subnets}, each specializing in a particular machine learning task. Since February 2025, the network employs Dynamic TAO (dTAO), under which each subnet $i$ maintains an independent constant-product AMM with reserves $(x_i, y_i)$, where $x_i$ is the TAO (native currency) reserve and $y_i$ is the subnet-specific ``alpha'' ($\alpha_i$) reserve \citep{bittensor2025dtao}.

Users ``stake'' TAO into a subnet by swapping TAO for alpha through the AMM, and ``unstake'' by swapping alpha back for TAO. The alpha price in TAO is $P_i = x_i / y_i$. Three features distinguish this setting from standard AMMs:

\begin{enumerate}[leftmargin=*]
    \item \textbf{No external market.} Alpha tokens trade exclusively through the on-chain AMM. There is no order book, no off-chain market, and no external price oracle. The AMM is the sole price discovery mechanism.

    \item \textbf{Pool liquidity injection.} Each block (approximately every 12 seconds), the protocol injects TAO into the subnet's AMM reserve. The TAO allocated to subnet $i$ is
    \begin{equation}\label{eq:emission}
        e_{\mathrm{TAO},i} = E_{\text{block}} \cdot \frac{\max(S_i - L, 0)}{\sum_{j} \max(S_j - L, 0)},
    \end{equation}
    where $E_{\text{block}}$ is the total TAO block emission (currently 0.5 TAO per block, following the December 2025 halving from 1 TAO per block), $S_i$ is the exponentially weighted moving average of net TAO flows into subnet $i$, and $L$ is a lower threshold. Simultaneously, alpha is injected into the pool in proportion $\Delta\alpha_i = \Delta\tau_i / P_i$, preserving the current spot price while deepening liquidity \citep{bittensor2025dtao}. This grows the invariant $k_i = x_i \cdot y_i$ over time without changing the price, mechanically dampening price volatility.

    \item \textbf{Alpha participant emissions.} Independently of the pool injection, each subnet emits alpha to \emph{participants} at a base rate of approximately 1 alpha per block, subject to its own halving schedule (both TAO and each alpha are capped at 21 million). At the end of each tempo (360 blocks), this participant alpha is distributed: 41\% to miners, 41\% to validators and their stakers, and 18\% to the subnet owner. This alpha does not enter the pool. It increases circulating supply outside the pool and can exert selling pressure if recipients swap their alpha back to TAO through the AMM.
\end{enumerate}

\begin{example}[Emission mechanics]\label{ex:emissions}
Suppose subnet $i$ has reserves $x_i = 1{,}000$ TAO and $y_i = 40{,}000$ ALPHA, so $P_i = 0.025$ TAO per ALPHA. In a single block, the pool injection channel adds $\Delta\tau_i = 0.01$ TAO\footnote{The total emission is 0.5 TAO per block (post-halving). With roughly 60 active subnets competing via \eqref{eq:emission}, 0.01 TAO per block is illustrative of an average subnet.} and $\Delta\alpha_i = 0.01/0.025 = 0.4$ ALPHA to the reserves. The new reserves are $x_i = 1{,}000.01$ and $y_i = 40{,}000.4$, the invariant grows from $k = 4 \times 10^7$ to $k' \approx 4.0001 \times 10^7$, and the price is unchanged at $P_i = 0.025$. Separately, the protocol emits 1 ALPHA to participants (0.41 to miners, 0.41 to validators and stakers, 0.18 to the subnet owner). This alpha does not enter the pool but increases circulating supply. If recipients sell it through the AMM, it exerts downward pressure on the alpha price. Over one tempo (360 blocks, roughly 72 minutes), the pool injection deepens reserves by 3.6 TAO and 144 ALPHA, while 360 ALPHA is distributed to participants.
\end{example}

\section{The Model}\label{sec:model}

\subsection{Setup and Notation}\label{sec:setup}

Consider a single AMM pool with reserves $(x(t), y(t))$ satisfying the constant-weighted-product invariant \eqref{eq:cwp}. Let $P(t) = \frac{1-w}{w} \cdot \frac{x(t)}{y(t)}$ denote the marginal price at time $t$. Define the \emph{net flow process} $F(t)$ as the cumulative net TAO staked into the pool by time $t$.

\begin{definition}[Stochastic Flow Process]\label{def:flow}
The net flow process satisfies
\begin{equation}\label{eq:flow}
    \diff F(t) = \mu_F \diff t + \sigma_F \diff W(t),
\end{equation}
where $\mu_F \in \R$ is the drift (expected net inflow rate), $\sigma_F > 0$ is the flow volatility, and $W(t)$ is a standard Brownian motion on a filtered probability space $(\Omega, \mathcal{F}, \{\mathcal{F}_t\}, \mathbb{P})$.
\end{definition}

The assumption that flow follows a Brownian diffusion is a continuous-time approximation to discrete staking and unstaking events. When individual staking amounts are small relative to pool size (the standard ``many small traders'' assumption), this is justified by the functional central limit theorem. We relax this assumption in \Cref{sec:discussion} by considering jump-diffusion flows.

\subsection{Reserve Dynamics}

In the absence of emissions, the reserves evolve according to the AMM mechanics:
\begin{equation}\label{eq:reserve_x}
    \diff x(t) = \diff F(t),
\end{equation}
and the constant-weighted-product constraint \eqref{eq:cwp} determines $y(t)$ implicitly:
\begin{equation}\label{eq:reserve_y}
    y(t) = \left(\frac{K}{x(t)^w}\right)^{1/(1-w)}.
\end{equation}
Differentiating via It\^{o}'s lemma yields
\begin{equation}\label{eq:dy}
    \diff y = -\frac{w}{1-w} \cdot \frac{y}{x} \diff x + \frac{w(2w-1)}{2(1-w)^2} \cdot \frac{y}{x^2} (\diff x)^2.
\end{equation}

\subsection{Derivation of the Price Process}

We now derive the central result: the stochastic differential equation governing the AMM token price.

\begin{theorem}[AMM Token Price Process]\label{thm:main}
Under the constant-weighted-product AMM \eqref{eq:cwp} with net flow process \eqref{eq:flow}, the marginal token price $P(t)$ satisfies the CEV stochastic differential equation
\begin{equation}\label{eq:cev}
    \diff P = \mu(P)\diff t + \delta P^w \diff W(t),
\end{equation}
where the CEV exponent is $\beta = w$ (the numeraire weight), the volatility parameter is
\begin{equation}\label{eq:delta}
    \delta = \frac{1}{1-w} \left(\frac{1-w}{w}\right)^{1-w} K^{-1} \sigma_F,
\end{equation}
and the drift is
\begin{equation}\label{eq:drift}
    \mu(P) = \frac{1}{1-w} \left(\frac{1-w}{w}\right)^{1-w} K^{-1} \mu_F P^{w} + \frac{w}{2(1-w)^2} \left(\frac{1-w}{w}\right)^{2(1-w)} K^{-2} \sigma_F^2 P^{2w-1}.
\end{equation}
\end{theorem}

\begin{proof}
From \eqref{eq:cwp}, the price $P = \frac{1-w}{w} \cdot \frac{x}{y}$ can be expressed purely as a function of $x$:
\begin{equation}\label{eq:P_of_x}
    P(x) = \frac{1-w}{w} \cdot x \cdot \left(\frac{x^w}{K}\right)^{1/(1-w)} = \frac{1-w}{w} \cdot K^{-1/(1-w)} \cdot x^{1/(1-w)}.
\end{equation}
Let $\alpha \equiv 1/(1-w)$ and $B \equiv \frac{1-w}{w} \cdot K^{-\alpha}$, so that $P = B x^\alpha$. By It\^{o}'s lemma,
\begin{align}
    \diff P &= B \alpha x^{\alpha-1} \diff x + \tfrac{1}{2} B \alpha(\alpha-1) x^{\alpha-2} (\diff x)^2 \notag \\
    &= B \alpha x^{\alpha-1} (\mu_F \diff t + \sigma_F \diff W) + \tfrac{1}{2} B \alpha(\alpha-1) x^{\alpha-2} \sigma_F^2 \diff t. \label{eq:dP_x}
\end{align}
Substituting $x = (P/B)^{1/\alpha}$:
\begin{align}
    x^{\alpha - 1} &= (P/B)^{(\alpha-1)/\alpha} = (P/B)^{1 - 1/\alpha} = (P/B)^w = B^{-w} P^w, \\
    x^{\alpha - 2} &= (P/B)^{(\alpha-2)/\alpha} = (P/B)^{2w-1} = B^{1-2w} P^{2w-1}.
\end{align}
Collecting terms, the diffusion coefficient of $\diff P$ is
\[
    B \alpha \cdot B^{-w} P^w \cdot \sigma_F = \alpha B^{1-w} \sigma_F \cdot P^w.
\]
Defining $\delta \equiv \alpha B^{1-w} \sigma_F$ and noting $B^{1-w} = \left(\frac{1-w}{w}\right)^{1-w} K^{-\alpha(1-w)} = \left(\frac{1-w}{w}\right)^{1-w} K^{-1}$, we obtain $\delta = \frac{1}{1-w}\left(\frac{1-w}{w}\right)^{1-w} K^{-1} \sigma_F$. The drift follows analogously.
\end{proof}

\begin{corollary}[Constant-Product AMM]\label{cor:cp}
For the standard constant-product AMM ($w = 1/2$, $K = \sqrt{k}$ where $k = xy$), the price process is
\begin{equation}\label{eq:cev_cp}
    \diff P = \left(\frac{2\mu_F}{\sqrt{k}} \sqrt{P} + \frac{\sigma_F^2}{k}\right) \diff t + \frac{2\sigma_F}{\sqrt{k}} \sqrt{P}\, \diff W(t).
\end{equation}
The CEV exponent is $\beta = 1/2$ and the volatility parameter is $\delta = 2\sigma_F / \sqrt{k}$.
\end{corollary}

\begin{corollary}[Black--Scholes Limit]\label{cor:bs_limit}
As pool depth $K \to \infty$, the volatility parameter $\delta \to 0$ and the price becomes deterministic. For large but finite $K$, with $P$ near $P_0$, the process approximates GBM:
\begin{equation}\label{eq:bs_approx}
    \frac{\diff P}{P} \approx \tilde{\mu}\, \diff t + \sigma_{\mathrm{eff}}\, \diff W(t), \qquad \sigma_{\mathrm{eff}} = \delta P_0^{w-1},
\end{equation}
and Black--Scholes applies with volatility $\sigma = \sigma_{\mathrm{eff}}$.
\end{corollary}

\begin{remark}[Elasticity spectrum]
The CEV exponent $\beta = w$ reveals a fundamental connection between AMM design and price dynamics:
\begin{itemize}
    \item $w = 1/2$ (constant-product): $\beta = 1/2$, variance decreases with price. This is the standard Uniswap/Bittensor case.
    \item $w \to 1$ (pool dominated by numeraire): $\beta \to 1$, approaching GBM and Black--Scholes.
    \item $w \to 0$ (pool dominated by alpha): $\beta \to 0$, approaching the Bachelier (normal) model.
\end{itemize}
AMM designers thus implicitly select a volatility structure through their choice of pool weights.
\end{remark}

\subsection{Properties of the CEV Price Process}\label{sec:properties}

\begin{proposition}[Volatility structure]\label{prop:vol}
Under \eqref{eq:cev}, the instantaneous return volatility is
\begin{equation}\label{eq:return_vol}
    \sigma_{\mathrm{ret}}(P) = \delta P^{w-1} = \delta P^{\beta - 1}.
\end{equation}
For the constant-product AMM ($\beta = 1/2$), $\sigma_{\mathrm{ret}}(P) = \delta / \sqrt{P}$: volatility is inversely proportional to the square root of price.
\end{proposition}

This property has a natural economic interpretation. When the alpha token price is high, the TAO reserve $x$ is large (since $P \propto x^2/k$), meaning the pool is deep in TAO terms. A given staking flow $\diff F$ then produces a smaller proportional change in $x$, hence a smaller proportional price impact. Conversely, when the price is low, the TAO reserve is shallow, and the same flow produces larger price swings.

\begin{proposition}[Implied volatility skew]\label{prop:skew}
The CEV model with $\beta < 1$ generates a negative implied volatility skew: out-of-the-money puts have higher Black--Scholes implied volatility than out-of-the-money calls. When implied volatilities are normalized by the ATM level, the skew shape depends only on $\beta$, not on the volatility parameter $\delta$ or pool depth $k$.
\end{proposition}

\begin{proof}
The negative skew is a standard property of the CEV model with $\beta < 1$ \citep{cox1996constant,davydov2003pricing}. For the universality of the normalized skew, observe the following. First, $C(\lambda P, \lambda K_{\mathrm{str}}, T) = \lambda C(P, K_{\mathrm{str}}, T)$ for all $\lambda > 0$ (homogeneity of degree one in spot and strike), so implied volatility depends only on the moneyness ratio $K_{\mathrm{str}}/P$. Second, the parameters of the non-central chi-squared distribution are $\kappa = 2r / [\delta^2 (1-\beta)(e^{2r(1-\beta)T} - 1)]$, $a = \kappa K_{\mathrm{str}}^{2(1-\beta)}$, $c = \kappa P^{2(1-\beta)} e^{2r(1-\beta)T}$, and $b = 1/(1-\beta)$. Since $\kappa \propto \delta^{-2}$, changing $\delta$ (equivalently, changing pool depth $k$) rescales both $a$ and $c$ by the same multiplicative factor, while the ratio $a/c = (K_{\mathrm{str}}/P)^{2(1-\beta)} e^{-2r(1-\beta)T}$ depends only on moneyness, $\beta$, and calendar parameters. The degrees of freedom $b$ depend only on $\beta$. The CEV call price, and hence the implied volatility at each moneyness, is therefore determined by $\beta$ once the ATM level is fixed. It follows that the \emph{normalized} skew $\sigma_{\mathrm{IV}}(K)/\sigma_{\mathrm{ATM}}$ is invariant to $\delta$ and $k$.
\end{proof}

\begin{proposition}[Boundary behavior at zero]\label{prop:boundary}
For $\beta = 1/2$, the CEV process \eqref{eq:cev} has $P = 0$ as a boundary point. The classification depends on the drift:
\begin{enumerate}[leftmargin=*]
    \item Under the risk-neutral measure $\Q$ (drift $rP$ with $r > 0$), the Feller test shows that the scale function $s(P) = \int^P \exp\!\left(-2r/(\delta^2 u)\right)\diff u$ diverges as $P \to 0^+$, so zero is an \emph{inaccessible} (entrance) boundary: the process cannot reach zero in finite time, and the pricing formula \eqref{eq:call} is well-defined without boundary correction.
    \item Under the physical measure $\mathbb{P}$, when the drift $\mu(P)$ is small relative to $\delta^2$ (i.e., $2\mu_F / (\sigma_F^2 \sqrt{k}) < 1$), the speed measure is integrable near zero and the boundary is \emph{attainable}: a sufficiently unfavorable sequence of outflows can drain the TAO reserve to zero.
\end{enumerate}
Economically, $P = 0$ corresponds to a fully drained TAO reserve ($x = 0$), at which point the AMM cannot quote a price. Once the TAO reserve is exhausted, no further trades are possible without external injection (e.g., emissions). For option pricing, the inaccessibility of zero under $\Q$ ensures that the non-central chi-squared formula accounts correctly for the boundary.
\end{proposition}

\begin{remark}[Leverage effect]\label{rem:leverage}
The CEV structure with $\beta < 1$ generates a negative correlation between the price level and return volatility: as $P$ falls, $\sigma_{\mathrm{ret}}(P) = \delta P^{\beta - 1}$ rises. In equity markets, this ``leverage effect'' is typically attributed to increased financial leverage as firm value declines \citep{black1973pricing}. For AMM tokens, the mechanism is purely structural. When the alpha price falls, the TAO reserve $x = \sqrt{kP}$ decreases, making the pool shallower in TAO terms. The same dollar-equivalent staking flow then moves the price proportionally more. The AMM bonding curve thus provides a first-principles derivation of the leverage effect, grounded in market microstructure rather than capital structure.
\end{remark}

\section{Option Pricing}\label{sec:pricing}

\subsection{Risk-Neutral Dynamics}\label{sec:risk_neutral}

Under the risk-neutral measure $\Q$, the CEV price process becomes
\begin{equation}\label{eq:cev_Q}
    \diff P = r P \diff t + \delta P^{\beta} \diff W^{\Q}(t),
\end{equation}
where $r$ is the risk-free rate and $W^{\Q}$ is a $\Q$-Brownian motion. The change of measure from $\mathbb{P}$ to $\Q$ is effected via the Girsanov kernel $\theta(P) = [\mu(P) - rP]/(\delta P^{\beta})$, which is bounded whenever $P$ is bounded away from zero. A standard localization argument establishes $\Q$: define stopping times $\tau_n = \inf\{t : P(t) < 1/n\}$ and apply Girsanov's theorem on each $[0, \tau_n \wedge T]$, where $\theta$ is bounded. Since $P(0) > 0$ and zero is inaccessible under the resulting risk-neutral dynamics (\Cref{prop:boundary}), $\tau_n \to \infty$ a.s.\ under $\Q$, yielding the global measure change. The existence of the equivalent martingale measure for CEV processes with $\beta \in (0,1)$ is established rigorously by \citet{davydov2003pricing} (see their \S3).

The risk-neutral pricing argument requires approximate replicability of contingent claims by dynamic trading in the underlying. For AMM tokens, this is imperfect due to slippage: a trade of size $\Delta x$ incurs a price impact of order $\Delta x / x$. When individual hedge trades are small relative to the TAO reserve (i.e., $|\Delta x| \ll x$), the slippage cost is second-order and the replication error is bounded. We quantify this error in \Cref{sec:hedging_error} and show it scales as $k^{-2}$, becoming negligible for deep pools.

\subsection{European Option Pricing Formula}\label{sec:cev_formula}

The European call price under the CEV model was derived by \citet{cox1975notes} and refined by \citet{schroder1989computing}. For $\beta < 1$ (which includes the constant-product case $\beta = 1/2$), zero is inaccessible under the risk-neutral measure (\Cref{prop:boundary}), so the non-central chi-squared representation is well-defined and no boundary correction is required. The price of a European call with strike $K_{\text{str}}$ and maturity $T$ is:

\begin{theorem}[CEV Call Price]\label{thm:call}
\begin{equation}\label{eq:call}
    C(P, K_{\mathrm{str}}, T) = P \left[1 - \chi^2\!\left(a;\, b + 2,\, c\right)\right] - K_{\mathrm{str}} \ee^{-rT} \chi^2\!\left(c;\, b,\, a\right),
\end{equation}
where $\chi^2(x;\, n,\, \lambda)$ denotes the cumulative distribution function of the non-central chi-squared distribution with $n$ degrees of freedom and non-centrality parameter $\lambda$, and
\begin{align}
    \kappa &= \frac{2r}{\delta^2(1-\beta)(e^{2r(1-\beta)T} - 1)}, \label{eq:kappa} \\
    c &= \kappa P^{2(1-\beta)} e^{2r(1-\beta)T}, \label{eq:c} \\
    a &= \kappa K_{\mathrm{str}}^{2(1-\beta)}, \label{eq:a} \\
    b &= \frac{1}{1-\beta}. \label{eq:b}
\end{align}
\end{theorem}

For the constant-product AMM ($\beta = 1/2$), these simplify to $b = 2$ and
\begin{equation}\label{eq:kappa_cp}
    \kappa = \frac{2r}{\delta^2 (e^{rT} - 1) / 2} = \frac{4r}{\delta^2(e^{rT} - 1)}, \quad c = \kappa P e^{rT}, \quad a = \kappa K_{\mathrm{str}}.
\end{equation}

The European put price follows from put-call parity:
\begin{equation}\label{eq:put}
    \Pi(P, K_{\mathrm{str}}, T) = C(P, K_{\mathrm{str}}, T) - P + K_{\mathrm{str}} \ee^{-rT}.
\end{equation}

\begin{remark}[Convergence to Black--Scholes]
As $\beta \to 1$ with $\delta P_0^{\beta - 1} = \sigma$ held fixed, the CEV call price \eqref{eq:call} converges to the Black--Scholes formula
\begin{equation}
    C_{\mathrm{BS}} = P\, \Phi(d_1) - K_{\mathrm{str}} \ee^{-rT} \Phi(d_2),
\end{equation}
where $d_{1,2} = \frac{\ln(P/K_{\mathrm{str}}) + (r \pm \sigma^2/2)T}{\sigma\sqrt{T}}$. The CEV formula thus formally justifies using Black--Scholes when AMM liquidity is large.
\end{remark}

\subsection{The Liquidity-Adjusted Black--Scholes Formula}\label{sec:liquidity_adjusted}

To make the connection to Black--Scholes explicit, we decompose the CEV call price as
\begin{equation}\label{eq:decompose}
    C_{\mathrm{CEV}} = C_{\mathrm{BS}}(\sigma_{\mathrm{eff}}) + \Lambda_C,
\end{equation}
where $\sigma_{\mathrm{eff}} = \delta P^{w-1}$ is the effective volatility at the current price, and $\Lambda_C$ is the \emph{liquidity correction}, i.e., the residual difference due to the price-dependent volatility structure.

\begin{proposition}[Liquidity correction]\label{prop:correction}
The liquidity correction $\Lambda_C = C_{\mathrm{CEV}} - C_{\mathrm{BS}}(\sigma_{\mathrm{eff}})$ at the money satisfies $\Lambda_C = O(\delta^2)$ as $\delta \to 0$ (equivalently, $k \to \infty$). For general moneyness, $\Lambda_C$ is positive for in-the-money calls (and out-of-the-money puts) and negative for out-of-the-money calls (and in-the-money puts), consistent with the negative skew generated by $\beta < 1$. At the money, the correction is positive but very small: the CEV call price slightly exceeds Black--Scholes.
\end{proposition}

The magnitude of the liquidity correction scales with $\delta^2 \propto K^{-2}$, confirming that it vanishes rapidly as pool depth increases.

\subsection{Liquidity-Adjusted Greeks}\label{sec:greeks}

The standard option Greeks are modified by the CEV structure. We also introduce two new sensitivities specific to AMM tokens.

\begin{definition}[AMM Greeks]\label{def:greeks}
The \emph{CEV delta} and \emph{CEV gamma} of a European call are
\begin{align}
    \Delta_{\mathrm{CEV}} &= \frac{\partial C}{\partial P} = 1 - \chi^2(a;\, b+2,\, c) + P \frac{\partial}{\partial P}\left[1 - \chi^2(a;\, b+2,\, c)\right] - K_{\mathrm{str}} e^{-rT} \frac{\partial}{\partial P} \chi^2(c;\, b,\, a), \label{eq:delta_cev} \\
    \Gamma_{\mathrm{CEV}} &= \frac{\partial^2 C}{\partial P^2}. \label{eq:gamma_cev}
\end{align}
The \emph{liquidity Greek} is
\begin{equation}\label{eq:lambda}
    \Lambda = \frac{\partial C}{\partial k} = \frac{\partial C}{\partial \delta} \cdot \frac{\partial \delta}{\partial k},
\end{equation}
measuring the sensitivity of the option price to changes in pool depth. For the constant-product AMM, $\delta = 2\sigma_F / \sqrt{k}$, so $\partial \delta / \partial k = -\sigma_F / k^{3/2}$, and $\Lambda < 0$: deeper pools reduce option value by compressing volatility.

The \emph{emission Greek} is
\begin{equation}\label{eq:emission_greek}
    \mathcal{E} = \frac{\partial C}{\partial e},
\end{equation}
measuring sensitivity to the emission rate $e$ that governs the growth of $k$ over time (see \Cref{sec:emissions}). Using the integrated variance \eqref{eq:int_var}, $\mathcal{E}$ can be computed via the chain rule: $\mathcal{E} = (\partial C / \partial \bar{v}^2)(\partial \bar{v}^2 / \partial \dot{k})(\partial \dot{k} / \partial e)$. Since $\partial \bar{v}^2 / \partial \dot{k} < 0$ and $\partial C / \partial \bar{v}^2 > 0$ (option prices increase with variance), the emission Greek is negative: higher emissions reduce option value by deepening the pool over the option's lifetime.
\end{definition}

\subsection{Hedging Error from AMM Friction}\label{sec:hedging_error}

Delta-hedging an option on an AMM token requires trading through the AMM, incurring slippage. For a hedge trade of $\Delta P$ units of alpha, the slippage cost in a constant-product AMM is approximately
\begin{equation}\label{eq:slippage}
    S(\Delta P) \approx \frac{(\Delta P)^2}{2k / P^2} = \frac{P^2 (\Delta P)^2}{2k}.
\end{equation}
Over a hedging interval $\Delta t$ with rebalancing, the cumulative expected hedging cost is
\begin{equation}\label{eq:hedge_cost}
    \E\left[\int_0^T S(\diff \Delta)\right] \approx \frac{P^2 \Gamma^2 \sigma_{\mathrm{ret}}^2}{2k} \cdot T,
\end{equation}
which scales as $k^{-2}$, an additional friction cost beyond the standard model. This cost should be added to the option price as a \emph{replication premium}.

\begin{proposition}[Replication premium]\label{prop:replication}
The replication premium for a European call on a constant-product AMM token is bounded by
\begin{equation}\label{eq:rep_premium}
    R \leq \frac{\delta^2 P^{2\beta}}{2k} \cdot \E^{\Q}\!\left[\int_0^T \Gamma_{\mathrm{CEV}}^2(t) P(t)^2 \diff t\right].
\end{equation}
As $k \to \infty$, $R \to 0$ and exact replication is recovered.
\end{proposition}

\begin{remark}[Replication--relevance tension]
A conceptual tension arises: the CEV correction to Black--Scholes is largest for shallow pools (small $k$), but the replication premium is also largest for small $k$ since it scales as $k^{-2}$. For the shallowest pool in our sample (SN58, $k = 7.4 \times 10^9$), the bound \eqref{eq:rep_premium} evaluates to less than $10^{-6}\%$ of the option price, confirming that the CEV-BS pricing discrepancy (of order $1$--$6\%$ of implied volatility) dominates the replication friction by many orders of magnitude. The $k^{-2}$ scaling of the friction versus the $k^{-1}$ scaling of the pricing discrepancy means the replication argument holds in precisely the regime where the CEV correction matters. For very shallow pools not represented in our data (e.g., $k < 10^5$), the replication premium could become material, and CEV prices should be interpreted as fair value benchmarks rather than strict arbitrage-free prices.
\end{remark}

\subsection{Extension: Token Emissions}\label{sec:emissions}

When emissions inject TAO and alpha into the pool at rates $\eTAO$ and $e_\alpha$ per unit time, the pool invariant grows deterministically:
\begin{equation}\label{eq:k_growth}
    \frac{\diff k}{\diff t} = y(t)\, \eTAO + x(t)\, e_\alpha.
\end{equation}
At equilibrium with $P$ near $P_0$, reserves grow approximately in proportion ($\Delta x / x \approx \Delta y / y$), maintaining roughly constant price while deepening the pool. For option horizons short relative to the emission timescale ($T \ll k_0 / \dot{k}$), this simplifies to
\begin{equation}\label{eq:k_growth_approx}
    k(t) \approx k_0 + (y_0\, \eTAO + x_0\, e_\alpha)\, t \equiv k_0 + \dot{k} \cdot t.
\end{equation}
The CEV volatility parameter becomes time-dependent:
\begin{equation}\label{eq:delta_t}
    \delta(t) = \frac{2\sigma_F}{\sqrt{k(t)}} = \frac{2\sigma_F}{\sqrt{k_0 + \dot{k} t}}.
\end{equation}

\begin{proposition}[Option pricing with emissions]\label{prop:emissions}
Under deterministically time-varying $\delta(t)$ and the assumption that $k(t)$ evolves slowly relative to the option horizon (so that the time-change from calendar time to ``variance time'' is approximately deterministic), the CEV call price formula \eqref{eq:call} remains valid with $\delta^2 T$ replaced by the integrated variance:
\begin{equation}\label{eq:int_var}
    \bar{v}^2 = \int_0^T \delta(t)^2 \diff t = 4\sigma_F^2 \int_0^T \frac{\diff t}{k_0 + \dot{k} t} = \frac{4\sigma_F^2}{\dot{k}} \ln\!\left(1 + \frac{\dot{k} T}{k_0}\right).
\end{equation}
As $\dot{k} \to 0$ (no emissions), $\bar{v}^2 \to 4\sigma_F^2 T / k_0 = \delta_0^2 T$, recovering the constant case.
\end{proposition}

Emissions have a dampening effect: higher emission rates increase $\dot{k}$, reducing $\bar{v}^2$ and hence option prices. Intuitively, growing liquidity compresses the range of possible price outcomes.

\begin{remark}[Emissions as a dividend yield]
The emission effect can be interpreted as an effective continuous dividend yield. Expanding $\bar{v}^2$ for small $\dot{k}T/k_0$:
\[
    \bar{v}^2 \approx \delta_0^2 T \left(1 - \frac{\dot{k}T}{2k_0} + \cdots\right),
\]
so the integrated variance is reduced by a factor that is first-order in $\dot{k}/k_0$. Defining an effective dividend yield $q_{\mathrm{eff}} = \dot{k}/(2k_0)$, the emission-adjusted CEV price approximately equals the zero-emission price on an underlying with continuous yield $q_{\mathrm{eff}}$. In practice, one can estimate $\dot{k}$ from the emission schedule and discount the option value accordingly.

Bittensor's emission allocation \eqref{eq:emission} introduces a feedback loop: high staking flows raise $S_i$, increasing the subnet's emission share, deepening the pool, and compressing volatility. Modeling this endogenous $\dot{k}$ coupled to the flow process is left for future work.
\end{remark}

\section{Numerical Analysis}\label{sec:numerical}

\subsection{Parameter Calibration}\label{sec:calibration}

We calibrate the model to Bittensor subnet data as of February 2026, retrieved from the Taostats API. We select three representative subnets spanning the range of pool depths observed across the network:

\begin{table}[H]
\centering
\caption{Representative Bittensor subnet parameters (median values over the sample period August 8, 2025 to February 23, 2026). Reserves are in native token units; $\hat{\sigma}_F$ is the annualized standard deviation of daily TAO reserve changes.}\label{tab:subnets}
\begin{tabular}{lrrrrr}
\toprule
Subnet & $x_0$ (TAO) & $y_0$ ($\alpha$) & $k$ ($\times 10^9$) & $P_0$ (TAO/$\alpha$) & $\hat{\sigma}_F$ \\
\midrule
Shallow (SN58) & 4{,}019 & 1{,}835{,}800 & 7.4 & 0.0022 & 2{,}293 \\
Medium (SN1)   & 22{,}568 & 2{,}340{,}831 & 52.8 & 0.0096 & 3{,}571 \\
Deep (SN3)     & 54{,}445 & 2{,}151{,}385 & 117.1 & 0.0253 & 8{,}250 \\
\bottomrule
\end{tabular}
\end{table}

The flow volatility $\hat{\sigma}_F$ is estimated from the standard deviation of daily net TAO reserve changes, annualized ($\hat{\sigma}_F = \hat{s}_{\Delta F} \cdot \sqrt{365}$, where $\hat{s}_{\Delta F}$ is the sample standard deviation of daily flow changes). This estimator is consistent under the diffusion assumption and can be computed directly from on-chain reserve data. The risk-free rate is set to $r = 5\%$ (approximate stablecoin lending rate in DeFi).

\begin{remark}[Distributional properties of staking flows]\label{rem:normality}
Shapiro--Wilk tests reject the normality of daily TAO reserve changes at the 5\% level for all 98 subnets in our sample. The median excess kurtosis is 10.7 (range $[0.6, 189.6]$) and the median skewness is $-0.62$, indicating heavy-tailed, left-skewed flow distributions. The Brownian diffusion assumption (\Cref{def:flow}) is thus an approximation, as is standard for continuous-time models applied to discrete financial data. The heavy tails are consistent with occasional large staking events (``whale'' trades) and support the jump-diffusion extension discussed in \Cref{sec:discussion}. That the cross-sectional backtest (\Cref{sec:backtest}) produces reasonable hedging errors (median MAE $\approx 3$\% of spot) despite these departures suggests the CEV framework is robust to moderate violations of the diffusion assumption.
\end{remark}

\begin{remark}[Testable predictions]
Even in the absence of traded options, the CEV model generates testable predictions about the physical price process. The return variance should be proportional to $P^{2(\beta-1)} = P^{-1}$ for the constant-product case, a relationship that can be estimated from realized variance regressions on price levels using on-chain data.
\end{remark}

\subsection{Monte Carlo Validation}\label{sec:mc}

To validate the closed-form formula, we simulate $N = 100{,}000$ paths of the flow process $\diff F = \mu_F \diff t + \sigma_F \diff W$ using Euler--Maruyama discretization with hourly time steps. At each step, the TAO reserve updates as $x_{t+\Delta t} = x_t + \Delta F_t$, the alpha reserve follows from the invariant $y = k/x$, and the terminal price is $P_T = x_T^2 / k$. The Monte Carlo call price is $C_{\mathrm{MC}} = e^{-rT} \E[\max(P_T - K_{\mathrm{str}}, 0)]$.

\Cref{fig:sample_paths} illustrates the qualitative difference between CEV and GBM dynamics. The left panel overlays paths from both models driven by the same Brownian increments: the paths visibly diverge as price moves away from $P_0$, with CEV producing wider swings at low prices and narrower swings at high prices. The right panel shows the terminal price distribution from 50{,}000 Monte Carlo paths, revealing the CEV model's heavier left tail, a direct consequence of the leverage effect.

\begin{figure}[H]
\centering
\includegraphics[width=\textwidth]{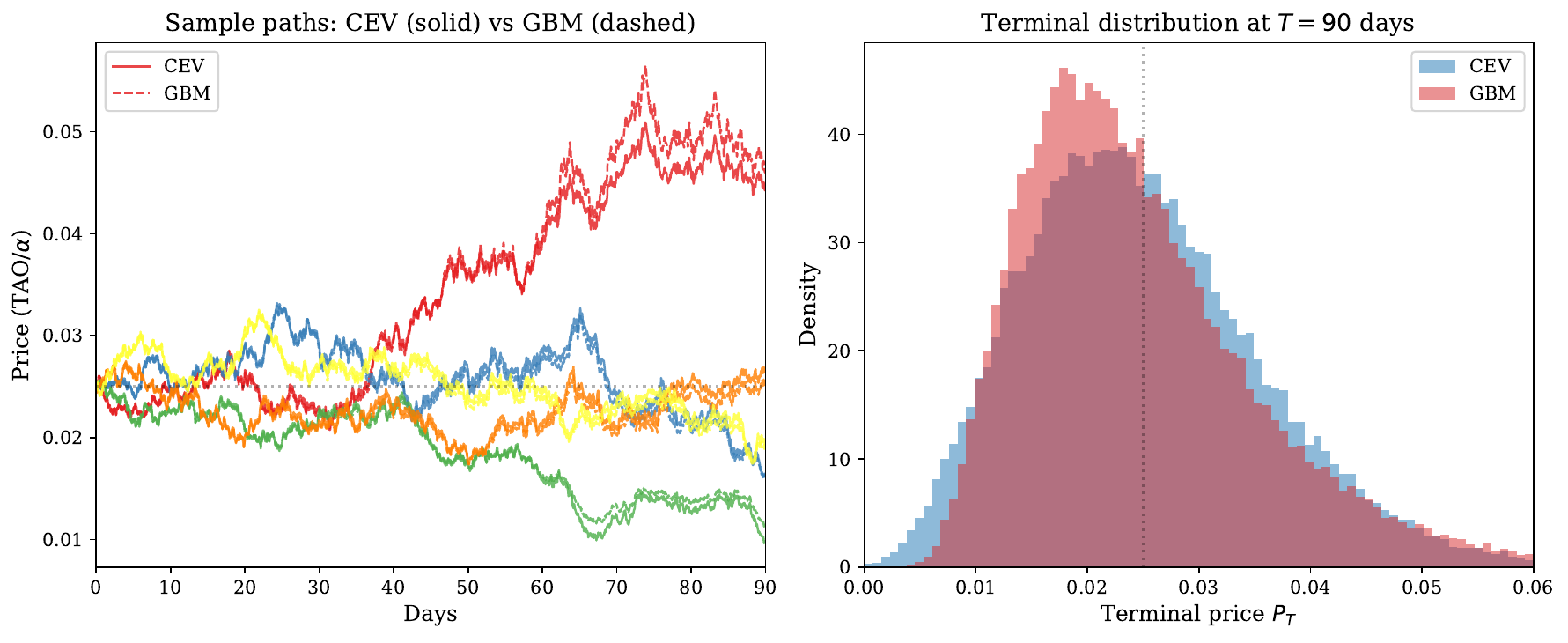}
\caption{Left: simulated price paths under CEV (solid) and GBM (dashed) driven by identical Brownian increments, with volatilities matched at $P_0$. The paths visibly diverge as price moves away from $P_0$: CEV produces wider swings at low prices (leverage effect) and narrower swings at high prices. Right: terminal price distribution from 50{,}000 Monte Carlo paths at $T = 90$ days. The CEV distribution exhibits a heavier left tail and positive skew relative to GBM, consistent with the structural leverage effect. Parameters: $P_0 = 0.025$, $k = 5 \times 10^5$, $\sigma_F = 48.7$.}
\label{fig:sample_paths}
\end{figure}

\Cref{fig:mc_validation} compares the Monte Carlo prices with the CEV closed-form formula across strikes for two pool depths ($T = 30$ days). For the deep pool ($k = 10^9$), the maximum deviation is less than 0.5\% of spot. For the shallow pool ($k = 10^6$), MC prices exceed the CEV formula by 1--3\% of spot: when $\sigma_F$ is large relative to the reserve $x_0 = \sqrt{kP_0}$, a single hourly flow shock can represent several percent of the reserve, violating the diffusion limit's small-increment assumption. The positive bias reflects truncation of large negative flow shocks at the reflecting boundary $x > 0$. The Monte Carlo standard error is below 0.2\% of spot for all strikes and pool depths.

\begin{figure}[H]
\centering
\includegraphics[width=0.85\textwidth]{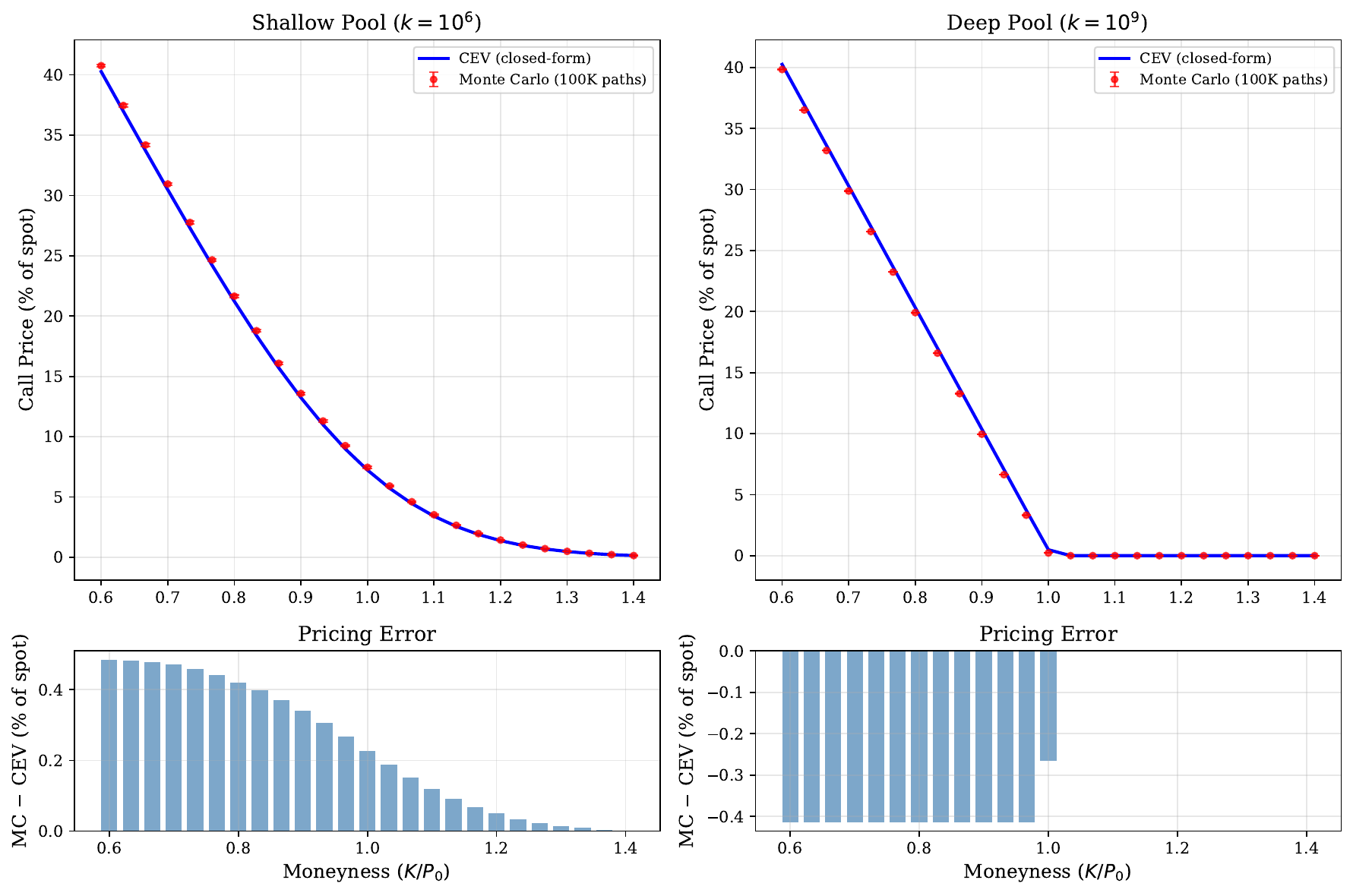}
\caption{Monte Carlo validation of the CEV pricing formula for a shallow pool ($k = 10^6$, left) and a deep pool ($k = 10^9$, right). Top: closed-form CEV call prices (line) vs.\ Monte Carlo estimates with 95\% confidence intervals (points). Bottom: pricing error (MC minus CEV) as a percentage of spot. The shallow pool shows a systematic positive bias of 1--3\% of spot, reflecting the Euler--Maruyama discretization error that is amplified when flow volatility is large relative to the reserve. The deep pool shows near-perfect agreement (errors $< 0.5$\%). Illustrative parameters: $P_0 = 0.025$, $\sigma_F = 48.7$, $T = 30$ days, 100,000 paths.}
\label{fig:mc_validation}
\end{figure}

\subsection{Comparative Statics}\label{sec:comp_statics}

\Cref{fig:price_vs_k} illustrates the relationship between option prices and pool depth. The left panel shows that CEV and Black--Scholes ATM call prices are visually indistinguishable. The right panel plots $|C_{\mathrm{CEV}} - C_{\mathrm{BS}}|$ on a log-log scale: OTM discrepancies dominate ATM by orders of magnitude, and all curves decline approximately as $O(k^{-1})$, confirming the scaling predicted by \Cref{prop:correction}.

\begin{figure}[H]
\centering
\includegraphics[width=\textwidth]{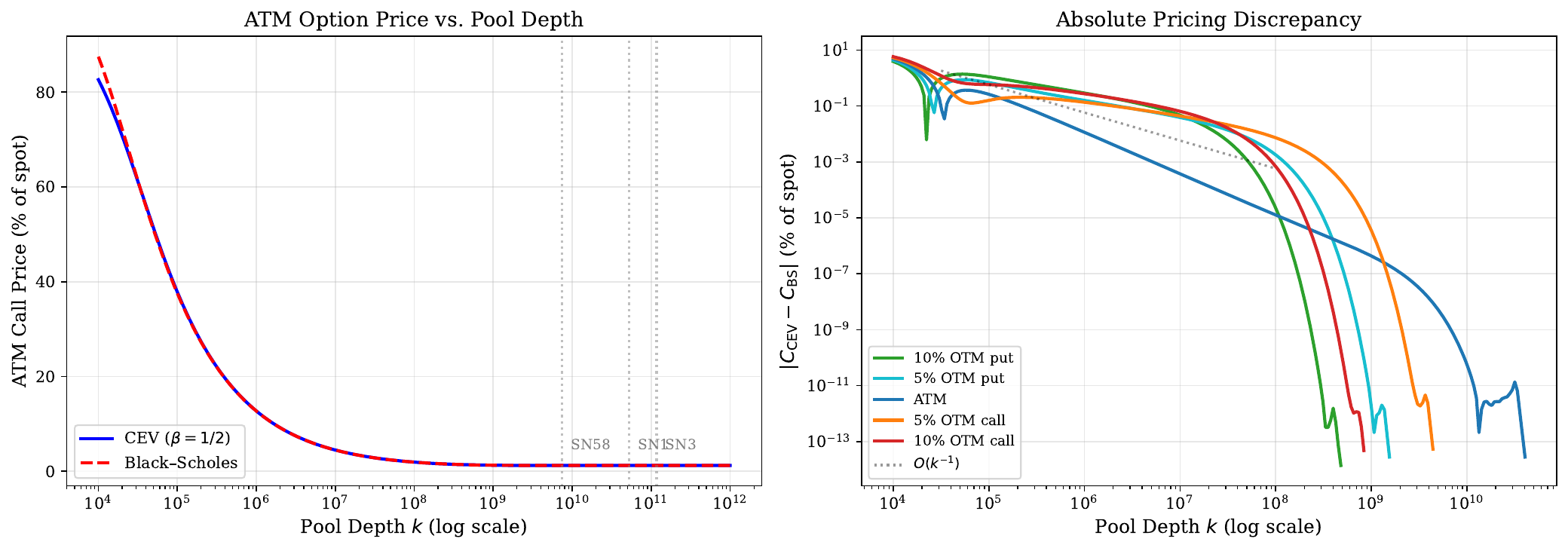}
\caption{Left: ATM call price (as \% of spot) vs.\ pool depth $k$. The CEV (solid blue) and Black--Scholes (dashed red) curves overlap, confirming that the models agree at the money. Right: absolute pricing discrepancy $|C_{\mathrm{CEV}} - C_{\mathrm{BS}}|$ (as \% of spot) on a log-log scale, for five moneyness levels. OTM options (puts in green/cyan, calls in orange/red) show discrepancies orders of magnitude larger than ATM (blue), reflecting the leverage-induced skew. All curves decline approximately as $O(k^{-1})$ (dotted reference line). Illustrative parameters: $P_0 = 0.025$, $\sigma_F = 48.7$, $T = 90$ days, $r = 5\%$.}
\label{fig:price_vs_k}
\end{figure}

\Cref{fig:vol_smile} plots implied volatility as a function of moneyness. Absolute implied volatility is higher for shallower pools ($\sigma_{\mathrm{ATM}} \propto k^{-1/2}$), but the normalized curves overlap almost exactly, confirming the universal-skew result of \Cref{prop:skew}: the skew depends on $\beta$ alone, not on pool depth.

\begin{figure}[H]
\centering
\includegraphics[width=0.85\textwidth]{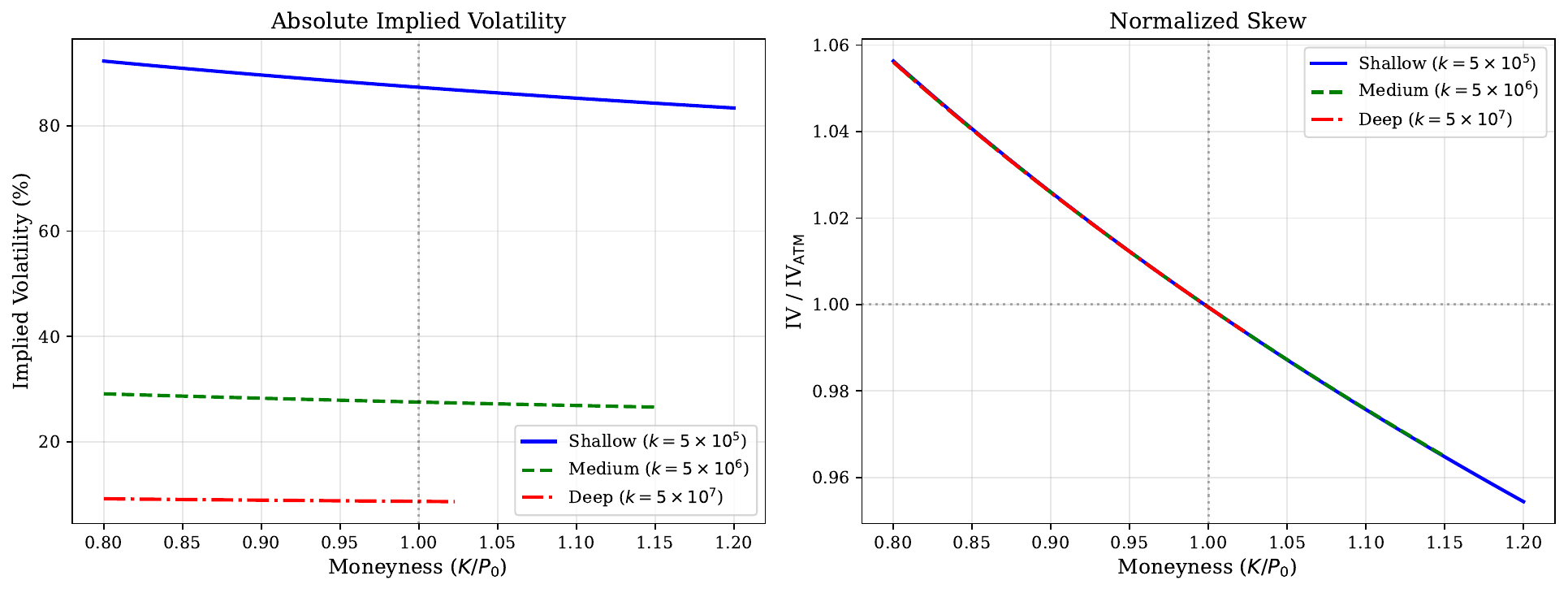}
\caption{Left: absolute Black--Scholes implied volatility extracted from CEV prices. Shallower pools produce higher absolute volatility due to larger $\delta$. Right: implied volatility normalized by the at-the-money level, isolating the skew shape. The three normalized curves overlap, confirming that the skew depends only on $\beta = 1/2$, not on pool depth. Illustrative parameters: $P_0 = 0.025$, $\sigma_F = 48.7$, $T = 90$ days.}
\label{fig:vol_smile}
\end{figure}

\Cref{fig:greeks} compares the CEV and Black--Scholes delta and gamma. The CEV delta is steeper for low prices and flatter for high prices, reflecting the price-dependent volatility. Both gammas peak below the strike, but the CEV gamma is sharper and peaks further below, concentrating risk in the high-volatility (low-price) region.

\begin{figure}[htbp]
\centering
\includegraphics[width=\textwidth]{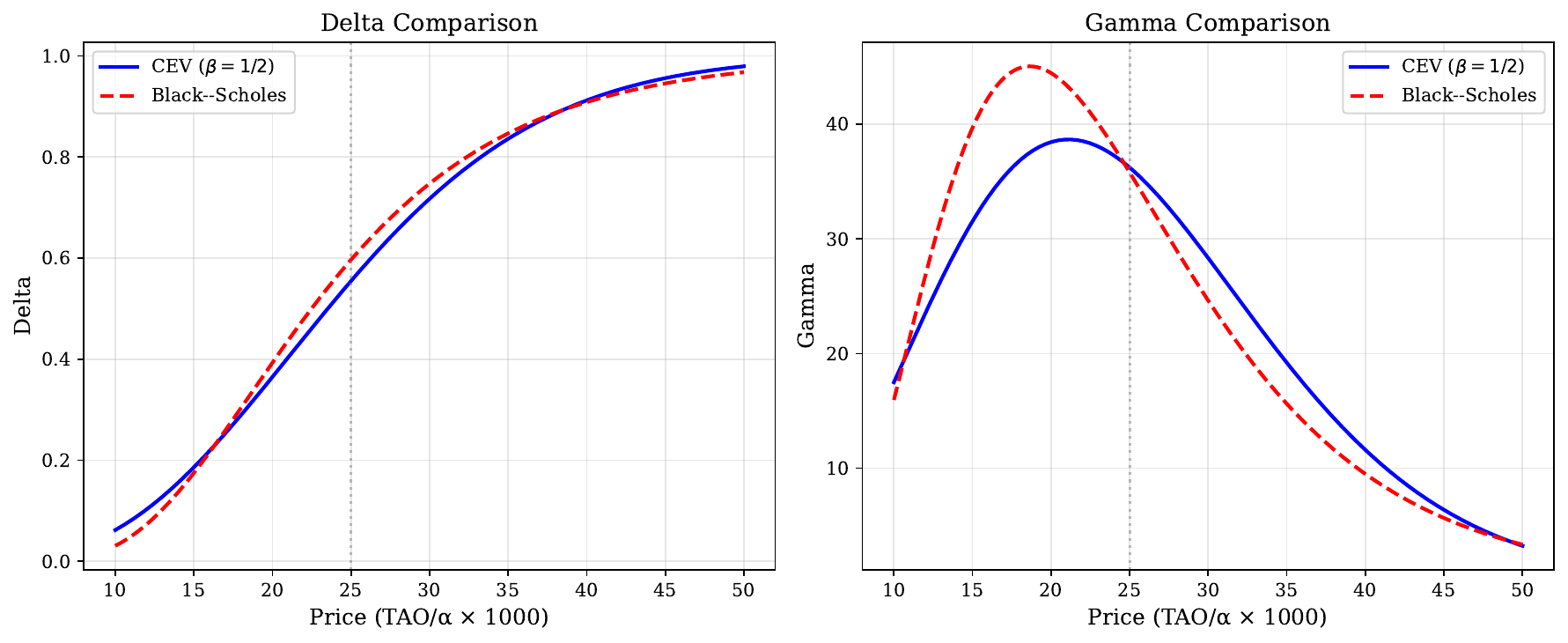}
\caption{Comparison of CEV ($\beta = 1/2$) and Black--Scholes Greeks for an ATM European call ($K = 0.025$) on a shallow pool. Left: delta. Right: gamma. Both gammas peak below the strike, but the CEV gamma is sharper and peaks further below, reflecting its concentration in the high-volatility (low-price) region. Illustrative parameters: $k = 5 \times 10^5$, $\sigma_F = 48.7$, $T = 90$ days.}
\label{fig:greeks}
\end{figure}

\subsection{Effect of Emissions}\label{sec:emissions_numerical}

\Cref{fig:emissions} shows that higher emissions compress option prices, particularly at longer maturities where cumulative liquidity deepening is greatest. \Cref{fig:liquidity_greek} plots the liquidity Greek $\Lambda = \partial C / \partial k$: it is negative throughout and decays rapidly, confirming that liquidity sensitivity is primarily a concern for shallow pools.

\begin{figure}[htbp]
\centering
\includegraphics[width=0.85\textwidth]{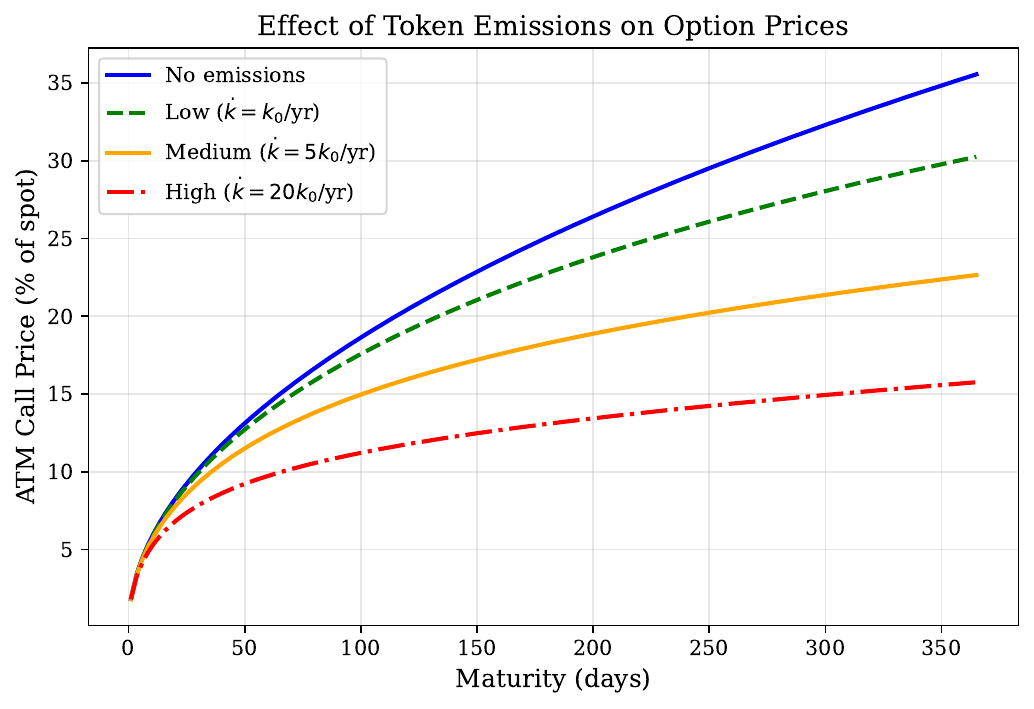}
\caption{ATM call option price (as \% of spot) vs.\ maturity under different emission rates, expressed as multiples of the initial pool invariant per year. Higher emissions deepen the pool over time, compressing volatility and reducing option prices at longer maturities. The effect is most pronounced for shallow pools with high emission-to-liquidity ratios. Illustrative parameters: $P_0 = K = 0.025$, $k_0 = 5 \times 10^5$, $\sigma_F = 48.7$.}
\label{fig:emissions}
\end{figure}

\begin{figure}[htbp]
\centering
\includegraphics[width=0.85\textwidth]{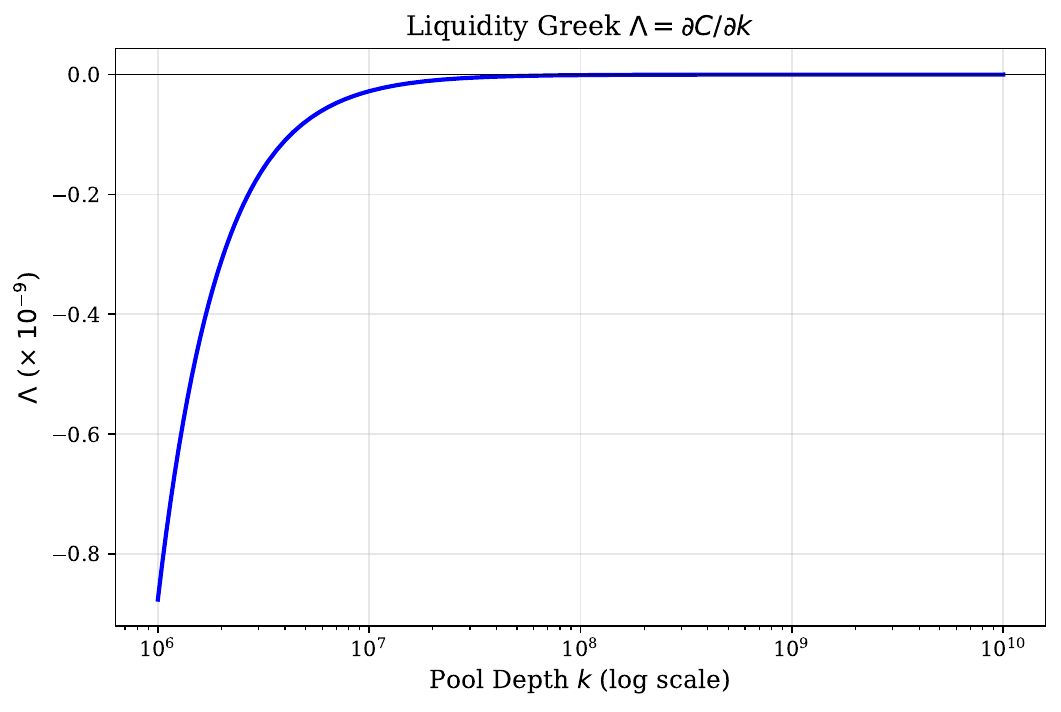}
\caption{The liquidity Greek $\Lambda = \partial C / \partial k$ for an ATM call as a function of pool depth. Negative values indicate that increasing pool depth reduces option value. The sensitivity is concentrated in shallow pools and becomes negligible for $k > 10^9$. Illustrative parameters: $P_0 = K = 0.025$, $\sigma_F = 48.7$, $T = 30$ days.}
\label{fig:liquidity_greek}
\end{figure}

\subsection{Empirical Backtest}\label{sec:backtest}

We conduct a cross-sectional backtest across all active Bittensor subnets to test whether the theoretical divergence between CEV and Black--Scholes pricing varies systematically with pool depth. Our sample construction proceeds as follows: of the 128 subnets in the network, 98 have sufficient on-chain history (at least 42 daily observations, covering a 14-day calibration window plus 14-day option horizon) for the backtest; of these 98, a subset of 90 additionally pass the stricter data-quality screens required for the variance elasticity test in \Cref{sec:variance_test}. The two samples overlap but are not nested, since the variance test applies different filters (degenerate price paths, minimum rolling-window observations) than the backtest's MAE filter. Using daily data from the 98 backtest-eligible subnets retrieved via the Taostats API (August 8, 2025 to February 23, 2026), we execute the following procedure for each subnet and each rolling start date $t$:
\begin{enumerate}[leftmargin=*]
    \item Estimate $\hat{\sigma}_F$ from the trailing 14-day standard deviation of daily TAO reserve changes, annualized.
    \item Compute the pool invariant $k_t = x_t \cdot y_t$ from observed reserves.
    \item Sell a 14-day ATM European call ($K = P_t$) at the model price under both the CEV model ($\beta = 1/2$, $\delta_t = 2\hat{\sigma}_F / \sqrt{k_t}$) and Black--Scholes with matched ATM volatility ($\sigma_{\mathrm{eff}} = \delta_t P_t^{-1/2}$).
    \item Delta-hedge daily for 14 days using each model's delta, updating $k$ and recomputing deltas from observed reserves at each rebalance.
    \item At expiry, compute the hedged P\&L: premium collected plus cumulative hedge gains minus the realized payoff $\max(P_{t+14} - K, 0)$.
\end{enumerate}
We aggregate each subnet's trades into a single mean absolute hedging error (MAE, as \% of spot) for each model. Because the 14-day option windows overlap, per-subnet MAE estimates exhibit serial correlation; we address this by using the cross-sectional regression (one observation per subnet), which is free of this overlap bias. After filtering 16 subnets with degenerate price paths (MAE $> 50$\%, typically from near-zero reserves or extreme price dislocations), 82 subnets remain.

\Cref{fig:backtest} presents the cross-sectional results. The left panel plots each subnet's CEV hedging error against its BS hedging error, with color indicating pool depth $\log_{10}(k)$. Points cluster tightly along the 45-degree line. The right panel quantifies the relationship between relative hedging performance and pool depth. An OLS regression of the CEV/BS error ratio on $\log_{10}(k)$ yields
\[
    \frac{\text{MAE}_{\text{CEV}}}{\text{MAE}_{\text{BS}}} = 1.029 - 0.002 \cdot \log_{10}(k), \quad R^2 = 0.003, \quad p = 0.60.
\]
The slope is not statistically significant. This is consistent with the theory: because the normalized implied volatility smile is universal for $\beta = 1/2$ (\Cref{prop:skew}), the CEV and BS ATM call prices are nearly identical at every pool depth, yielding near-identical deltas and hedging errors. The backtest confirms that for ATM options, the two models are empirically indistinguishable. The CEV correction matters most for out-of-the-money options, where the implied volatility skew generates meaningful pricing differences (\Cref{fig:price_vs_k}), but this effect cannot be tested without active OTM option markets on AMM tokens.

Only 17 of 82 subnets (21\%) show lower hedging error under CEV. At each rebalance, the BS hedge uses $\sigma_{\mathrm{eff}} = \delta P_t^{-1/2}$ (the CEV local volatility), ensuring a fair ATM comparison but mechanically limiting scope for divergence. Both models produce modest hedging errors (median MAE $\approx 3$\% of spot), confirming the diffusion approximation is reasonable for most subnets.

The MAE $> 50$\% filter could introduce selection bias if it disproportionately removes shallow pools, where CEV--BS divergence should be largest. The 16 excluded subnets are indeed shallower on average (median $\log_{10}(k) = 9.20$ vs.\ $9.93$ for the included subnets; see \Cref{tab:dropped_subnets} in the appendix). However, re-running the regression on all 98 subnets without any MAE filter yields a qualitatively identical result:
\[
    \frac{\text{MAE}_{\text{CEV}}}{\text{MAE}_{\text{BS}}} = 1.048 - 0.003 \cdot \log_{10}(k), \quad R^2 = 0.016, \quad p = 0.21.
\]
The slope remains negative and insignificant, confirming that the main finding---near-identical ATM hedging performance regardless of pool depth---is robust to the inclusion of degenerate subnets (\Cref{app:robustness}).

\begin{figure}[H]
\centering
\includegraphics[width=\textwidth]{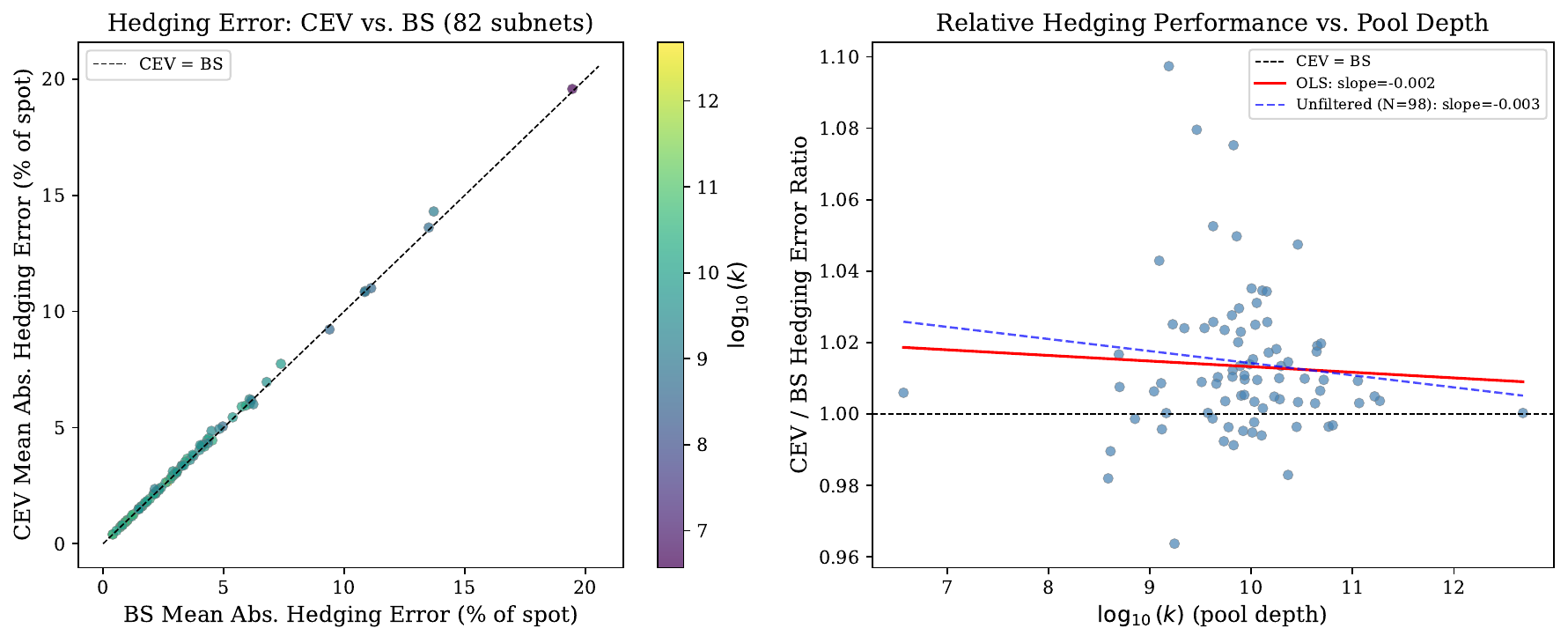}
\caption{Cross-sectional delta-hedged backtest of 14-day ATM calls across 82 Bittensor subnets (August 8, 2025 to February 23, 2026). Left: each subnet's mean absolute hedging error under CEV ($y$-axis) vs.\ BS ($x$-axis), colored by pool depth $\log_{10}(k)$. Points cluster tightly along the 45-degree line, confirming that the two models produce near-identical hedging errors at the money. Right: the CEV/BS error ratio shows no significant dependence on pool depth (OLS slope $= -0.002$, $p = 0.60$). Data source: Taostats API, daily pool snapshots.}
\label{fig:backtest}
\end{figure}

\subsection{Variance Elasticity Test}\label{sec:variance_test}

The hedging backtest in \Cref{sec:backtest} compares CEV and Black--Scholes at the money, where both models agree by construction. We now test a prediction that distinguishes the two models at all strikes: the CEV variance elasticity.

Under the CEV model, the instantaneous return variance is $\sigma_{\mathrm{ret}}^2(P) = \delta^2 P^{2(\beta-1)}$. Substituting $\delta = 2\sigma_F / \sqrt{k}$:
\begin{equation}\label{eq:var_elasticity}
    \sigma_{\mathrm{ret}}^2(P) = \frac{4\sigma_F^2}{k} \cdot P^{2(\beta-1)}.
\end{equation}
Taking logarithms and rearranging:
\begin{equation}\label{eq:var_test}
    \log\!\left(\widehat{\mathrm{RV}} \cdot k \,/\, \hat{\sigma}_F^2\right) = \mathrm{const} + 2(\beta - 1) \log P,
\end{equation}
where $\widehat{\mathrm{RV}}$ is the annualized realized variance of daily log returns in a rolling 14-day window (the sum of squared daily log returns multiplied by $365/14$). For $\beta = 1/2$, the slope is $-1$; for GBM ($\beta = 1$), the slope is $0$. The left-hand side of \eqref{eq:var_test} controls for both pool depth $k$ and flow volatility $\hat{\sigma}_F^2$ (the annualized sample variance of daily TAO reserve changes within the same window), isolating the pure price-variance relationship.

We estimate \eqref{eq:var_test} within each subnet that passes additional data-quality screens beyond the history requirement of \Cref{sec:backtest}: no degenerate price paths (price range exceeding $10^4 \times$, zero price variance, or non-positive reserves), and at least 10 valid 14-day rolling-window observations after removing $3\sigma$ outliers. The $3\sigma$ filter removes rolling-window observations whose log-adjusted realized variance lies more than three standard deviations from the subnet mean, reducing the influence of extreme jump days on the slope estimate. Of the 98 subnets with sufficient history, 90 survive these screens; the eight additional exclusions have extreme price-to-reserve ratios or insufficient intra-window variation to estimate the controlled regression reliably. The resulting 90 within-subnet slope estimates form a distribution. \Cref{fig:var_elasticity} presents the results. The median slope is $-0.86$ (interquartile range $[-0.98, -0.71]$), with 94\% of subnets showing negative slopes. We conduct two pre-specified one-sample $t$-tests: against the GBM null (slope $= 0$), which rejects decisively ($t = -11.3$, $p < 10^{-4}$); and against the exact CEV prediction (slope $= -1$; $t = 3.7$, $p < 0.001$). Both rejections survive Bonferroni correction for two hypotheses (adjusted $\alpha = 0.025$). The implied variance elasticity corresponds to $\hat{\beta} \approx 0.57$, modestly above the theoretical $1/2$ for constant-product AMMs. Three factors may contribute to this attenuation: (i) the discrete, jump-like nature of large staking events attenuates measured elasticity relative to the continuous-time prediction; (ii) measurement noise from overlapping 14-day rolling windows biases slope estimates toward zero; and (iii) some Bittensor pools may operate with effective weights slightly above $1/2$ due to the recent introduction of concentrated liquidity features. Disentangling these factors requires per-subnet estimation of effective pool weights, which we leave for future work.

The key finding is that the data strongly favor the CEV structure over GBM: the variance-price elasticity is robustly negative across subnets, as predicted by $\beta < 1$, and the magnitude is close to the $-1$ predicted by $\beta = 1/2$.

\begin{figure}[H]
\centering
\includegraphics[width=\textwidth]{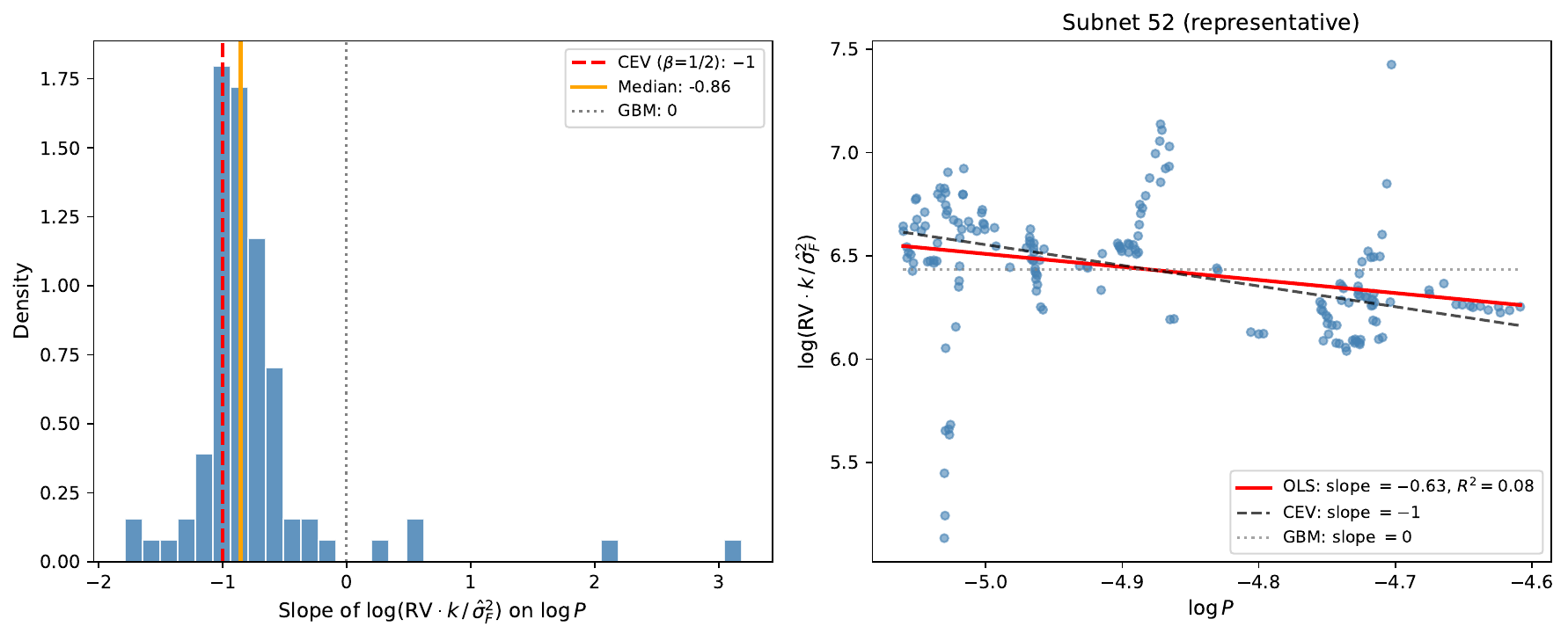}
\caption{Cross-sectional test of the CEV variance elasticity. Left: distribution of within-subnet slopes of $\log(\widehat{\mathrm{RV}} \cdot k / \hat{\sigma}_F^2)$ on $\log P$ across 90 Bittensor subnets. The median slope is $-0.86$, close to the CEV prediction of $-1$ (dashed red) and far from the GBM prediction of $0$ (dotted gray). Right: scatter plot for a representative subnet (SN52), illustrating the negative relationship between the controlled realized variance measure and price. Data source: Taostats API, August 8, 2025 to February 23, 2026.}
\label{fig:var_elasticity}
\end{figure}

\section{Discussion}\label{sec:discussion}

\subsection{Limitations}

First, most AMMs charge a swap fee (e.g., 0.3\% on Uniswap). Fees modify the effective invariant: a trade of $\Delta x$ yields $\Delta y = y(1-\phi)\Delta x / (x + (1-\phi)\Delta x)$, where $\phi$ is the fee rate. This introduces a bid-ask spread but does not alter the qualitative CEV structure; the main effect is to reduce the effective $\sigma_F$ by a factor of $(1-\phi)$ in the diffusion limit. Bittensor's dTAO pools currently charge no explicit swap fee, making our zero-fee model directly applicable.

Second, the diffusion assumption for staking flows is an approximation. In practice, large staking events (``whale'' transactions) can produce jump-like price movements. Across the 98 subnets in our sample, we identify jump days as those where the absolute daily TAO reserve change exceeds $3\hat{\sigma}_F$ of the trailing 14-day window. Such events occur on approximately 7.7\% of trading days (1{,}394 of 18{,}130 total subnet-days), but account for a median of 47\% of the total realized variance of reserve flows across subnets (interquartile range 31--65\%). Despite their outsized variance contribution, the cross-sectional hedging errors (\Cref{fig:backtest}) remain modest (median MAE $\approx 3$\% of spot), suggesting that most jumps are small relative to the cumulative diffusive variation over a 14-day hedging horizon. A Merton-type jump-diffusion extension \citep{merton1976option} would augment the flow process with a compound Poisson component $J_t \diff N_t$, leading to a jump-diffusion CEV model whose option pricing formula involves a weighted sum of CEV prices across possible jump scenarios. For subnets with frequent large staking events, such an extension may yield tighter hedging bounds, particularly for short-dated options where a single jump can dominate the realized path.

Third, the risk-neutral pricing argument requires the ability to delta-hedge, which is imperfect due to AMM slippage. Our replication premium (\Cref{prop:replication}) provides a bound on this friction, but a more rigorous treatment would employ the utility-based framework of \citet{davis1993european} for markets with transaction costs.

Fourth, the model treats flow volatility $\sigma_F$ as constant. In practice, staking activity exhibits time-of-day effects, momentum, and regime changes. A stochastic volatility extension that layers Heston-type dynamics \citep{heston1993closed} onto the flow process would capture these features at the cost of analytical tractability.

Fifth, staking flows may respond to the TAO/USD exchange rate, since TAO trades on centralized exchanges. For TAO-denominated derivatives, the CEV dynamics describe the alpha/TAO price conditional on a given flow process, so the framework remains valid. For USD-denominated derivatives, one would need to jointly model the TAO/USD price and the flow process, likely introducing stochastic correlation.

Sixth, AMM token prices are vulnerable to manipulation near option expiry. The cost of moving the price by a fraction $\epsilon$ is approximately $\epsilon \cdot x$ (the TAO reserve), which for shallow pools may be small relative to the option payoff gained. Practical implementations should incorporate safeguards such as time-weighted average prices (TWAPs) for settlement or oracle-based price feeds aggregated over multiple blocks.

\subsection{Extensions}

\emph{Concentrated liquidity.} Uniswap V3 restricts the constant-product invariant to a bounded price range $[P_a, P_b]$. Within this range, local dynamics are equivalent to a constant-product AMM with effective invariant $k_{\mathrm{eff}}$, so our CEV result applies locally; the framework of \citet{cartea2023decentralised} provides a natural starting point for this extension. At the range boundaries, the position becomes single-sided, which could be modeled as an absorbing or reflecting barrier. \emph{Cross-subnet options.} Correlating the Brownian motions across subnet staking flows enables pricing of basket or spread options on multiple alpha tokens. \emph{American and perpetual options.} American options can be priced via the free-boundary CEV formulation \citep{detemple2002american}; perpetual options \citep{dave2023perpetual} fit naturally through the stationary solution of the pricing PDE.

\subsection{Practical Implications}

The central practical implication is that Black--Scholes \emph{underprices} downside protection on AMM tokens at every pool depth. The leverage effect (\Cref{rem:leverage}) elevates implied volatility for OTM puts: as the token price falls toward the strike, the bonding curve amplifies volatility, making further declines more likely than the lognormal model predicts. Because the normalized skew depends on $\beta$ rather than $k$ (\Cref{fig:vol_smile}), this is a structural feature, not a small-pool artifact. Conversely, Black--Scholes overprices OTM calls.

As a concrete example, consider 90-day 20\%-out-of-the-money puts on three representative Bittensor subnets (\Cref{tab:subnets}). The CEV model prices these puts 10--28\% higher than Black--Scholes (with matched ATM volatility): 12.1\% vs.\ 11.1\% of spot for the shallow pool, 0.52\% vs.\ 0.41\% for the medium pool, and 0.43\% vs.\ 0.34\% for the deep pool. A subnet treasury holding 100{,}000 TAO worth of alpha tokens and seeking downside protection would pay 95--1{,}070 TAO more under the CEV model than Black--Scholes suggests, depending on pool depth. A market maker using Black--Scholes to sell these puts would systematically underprice the risk.

\section{Conclusion}\label{sec:conclusion}

We have shown that the price of a token traded on a constant-weighted-product automated market maker follows a constant elasticity of variance process, with the CEV exponent equal to the numeraire weight. This result is derived from first principles: given a diffusion model for staking flows, the AMM's bonding curve mechanics fully determine the price dynamics, with the CEV exponent pinned by pool design rather than estimated from data. The Black--Scholes model emerges as the limiting case of infinite pool depth, providing a precise characterization of when standard pricing tools are adequate and when they are not.

The framework yields closed-form European option prices via the non-central chi-squared distribution, AMM-specific liquidity and emission Greeks, and a quantitative decomposition of the pricing discrepancy relative to Black--Scholes as a function of pool depth and moneyness. The CEV structure also provides a first-principles derivation of the leverage effect for AMM tokens: the negative correlation between price and volatility arises directly from the bonding curve, making it a structural prediction rather than an empirical regularity.

A cross-sectional variance elasticity test across 90 subnets provides direct evidence for the CEV structure: after controlling for pool depth and flow volatility, realized return variance scales as $P^{-0.86}$, strongly rejecting the GBM null of $P^0$ ($p < 10^{-4}$) and broadly consistent with the $P^{-1}$ predicted by $\beta = 1/2$. A complementary delta-hedged backtest of ATM calls across 82 subnets confirms near-identical hedging errors at the money ($p = 0.60$), consistent with the prediction that the CEV and Black--Scholes pricing discrepancy is concentrated in the wings. The backtest validates the diffusion approximation for most subnets (median hedging error $\approx 3\%$ of spot), and the variance elasticity test validates the CEV price-volatility relationship that drives the skew.

AMM-native tokens are proliferating across decentralized protocols, creating demand for derivative pricing tools tailored to these instruments. The CEV framework developed here provides a foundation for pricing, hedging, and risk management that respects the structural constraints of the underlying market mechanism.

\bigskip

\paragraph{Disclosure.} The author has no financial interest in Bittensor, TAO tokens, or any DeFi protocol discussed herein. Data were obtained from publicly available on-chain sources via the Taostats API. The author declares no conflicts of interest.

\bigskip

\bibliographystyle{plainnat}

\clearpage
\appendix

\section{Proofs}\label{app:proofs}

\subsection{Proof of \Cref{prop:correction}}

Write the CEV call price \eqref{eq:call} as $C(\beta)$, treating $\delta$ as a function of $\beta$ through the constraint $\sigma_{\mathrm{eff}} = \delta P^{\beta - 1} = \mathrm{const}$. Then $C(1) = C_{\mathrm{BS}}(\sigma_{\mathrm{eff}})$.

The parameters $a$, $b$, $c$ depend on $\beta$ through the exponents $2(1-\beta)$ and $b = 1/(1-\beta)$. As $\beta \to 1$, $b \to \infty$ and the non-central chi-squared distribution converges to a normal. The convergence rate is $O(1/b) = O(1-\beta)$, yielding $C(\beta) - C(1) = O(1-\beta) = O(\delta^2)$, since $\delta \propto (1-\beta)^{1/2}$ when $\sigma_{\mathrm{eff}}$ is held fixed.

To determine the sign at the money, note that the CEV model's conditional variance $\E[P_T^2 | P_0] - (\E[P_T | P_0])^2$ exceeds the GBM variance (by Jensen's inequality applied to the convex function $P \mapsto P^{2(\beta-1)}$ when $\beta < 1$). Since call prices are increasing in variance, $C_{\mathrm{CEV}} \geq C_{\mathrm{BS}}(\sigma_{\mathrm{eff}})$ at the money, with equality only when $\beta = 1$. For OTM calls, the sign reverses due to the skew: the CEV model assigns less probability mass to large upward moves, so $C_{\mathrm{CEV}} < C_{\mathrm{BS}}$ for sufficiently high strikes. This is consistent with the negative implied volatility skew (\Cref{prop:skew}).

\subsection{Proof of \Cref{prop:replication}}

The hedging error over $[t, t+\Delta t]$ from trading $\Delta_{\mathrm{CEV}} \cdot \Delta P$ units through the AMM with price impact \eqref{eq:slippage} is
\[
    \epsilon_t = S\!\left(\Gamma_{\mathrm{CEV}} (\Delta P)^2 / (2P)\right) \approx \frac{P^2 \Gamma_{\mathrm{CEV}}^2}{2k} \cdot \delta^2 P^{2\beta} \Delta t,
\]
using $(\Delta P)^2 \approx \delta^2 P^{2\beta} \Delta t$. Integrating and taking expectations under $\Q$ gives \eqref{eq:rep_premium}.

\subsection{Proof of \Cref{prop:emissions}}

For time-dependent $\delta(t)$, the CEV transition density depends on the total integrated variance $\bar{v}^2 = \int_0^T \delta(t)^2 \diff t$ when the time-change technique is applied. Substituting $\delta(t) = 2\sigma_F / \sqrt{k_0 + \dot{k}t}$ and evaluating the integral yields \eqref{eq:int_var}.

\section{Robustness of the MAE Filter}\label{app:robustness}

The cross-sectional backtest in \Cref{sec:backtest} excludes 16 subnets with mean absolute hedging error exceeding 50\% of spot. Because this filter could introduce selection bias---particularly if shallow pools are disproportionately excluded---we report the characteristics of the dropped subnets and a robustness check without any post-hoc filtering.

\Cref{tab:dropped_subnets} lists the 16 excluded subnets. All have extremely high MAE values (typically 1{,}000--20{,}000\% of spot), indicating degenerate price dynamics---most commonly a near-collapse to zero reserves followed by recovery, or an extreme price dislocation early in the sample. The excluded subnets are systematically shallower than the included ones: their median pool depth is $\log_{10}(k) = 9.20$, compared with $9.93$ for the 82 included subnets. This is expected, since shallow pools are more susceptible to the reserve depletions and extreme dislocations that generate degenerate hedging outcomes.

Re-running the OLS regression of the CEV/BS hedging error ratio on $\log_{10}(k)$ using all 98 subnets (without the MAE filter) yields:
\[
    \frac{\text{MAE}_{\text{CEV}}}{\text{MAE}_{\text{BS}}} = 1.048 - 0.003 \cdot \log_{10}(k), \quad R^2 = 0.016, \quad p = 0.21.
\]
The slope remains negative and statistically insignificant, confirming that the main finding---near-identical ATM hedging performance regardless of pool depth---is not an artifact of the sample restriction.

\begin{table}[H]
\centering
\caption{Characteristics of the 16 subnets excluded by the MAE $> 50$\% filter. All exhibit extremely high hedging errors indicative of degenerate price paths. The excluded subnets are systematically shallower than the 82 included subnets (median $\log_{10}(k) = 9.20$ vs.\ $9.93$).}\label{tab:dropped_subnets}
\begin{tabular}{rrrrr}
\toprule
Subnet & $\log_{10}(k)$ & CEV MAE (\%) & BS MAE (\%) & CEV/BS Ratio \\
\midrule
86  & 8.89 & 3{,}757  & 3{,}684  & 1.020 \\
126 & 8.98 & 7{,}456  & 7{,}145  & 1.044 \\
109 & 8.99 & 2{,}336  & 2{,}299  & 1.016 \\
108 & 9.01 & 1{,}019  & 993      & 1.026 \\
94  & 9.06 & 2{,}214  & 2{,}174  & 1.018 \\
49  & 9.07 & 3{,}957  & 3{,}796  & 1.042 \\
113 & 9.09 & 1{,}779  & 1{,}713  & 1.038 \\
99  & 9.19 & 2{,}342  & 2{,}321  & 1.009 \\
87  & 9.21 & 20{,}072 & 19{,}490 & 1.030 \\
114 & 9.23 & 2{,}081  & 2{,}050  & 1.015 \\
107 & 9.31 & 2{,}647  & 2{,}602  & 1.017 \\
80  & 9.37 & 5{,}620  & 5{,}439  & 1.033 \\
47  & 9.72 & 4{,}213  & 4{,}111  & 1.025 \\
38  & 9.75 & 4{,}363  & 4{,}299  & 1.015 \\
76  & 9.76 & 433{,}198 & 433{,}194 & 1.000 \\
15  & 9.82 & 4{,}228  & 4{,}193  & 1.008 \\
\bottomrule
\end{tabular}
\end{table}

\end{document}